\documentclass[11pt]{article}

\usepackage{booktabs}
\usepackage{tabularx}
\usepackage{lipsum}
\usepackage{cite}

\usepackage{amsthm}
\usepackage{mathtools}
\usepackage{amssymb}
\usepackage{amsmath}
\usepackage{algorithm}
\usepackage{algorithmic}
\usepackage[affil-it]{authblk}
\newtheorem{theorem}{Theorem}[section]
\newtheorem{lemma}[theorem]{Lemma}
\newtheorem{claim}[theorem]{Claim}
\newtheorem{fact}[theorem]{Fact}
\newtheorem{cor}[theorem]{Corollary}

\newcommand{\zone}{\bra{0, 1}}
\newcommand{\adeg}{\widetilde{\mathrm{deg}}}
\newcommand{\eps}{\epsilon}

\newcommand\R{\mathbb R}
\newcommand\dout{D^\text{out}}
\newcommand\din{D^\text{in}}

\providecommand\bigabs[1]{\bigl\lvert#1\bigr\rvert}

\newcommand{\eat}[1]{}

\newcommand{\bra}[1]{\{#1\}}

\newcommand{\mathify}[1]{\ifmmode{#1}\else\mbox{$#1$}\fi}
\newcommand{\abs}[1]{\mathify{\left| #1 \right|}}

\newcommand{\OR}{\mathsf{OR}}
\newcommand{\PARITY}{\mathsf{XOR}}

\newtheorem*{corollary*}{Corollary}

\theoremstyle{definition}

\DeclareMathOperator\E{\mathbb{E}}

\usepackage{dsfont}
\title{Sharp indistinguishability bounds from non-uniform approximations}

\author{Christopher Williamson%
  \thanks{\texttt{chris@cse.cuhk.edu.hk}}}

\begin{document}
\maketitle
\abstract
We study the problem of distinguishing between two symmetric probability distributions over $n$ bits by observing $k$ bits of a sample, subject to the constraint that all $k-1$-wise marginal distributions of the two distributions are identical to each other. Previous works of Bogdanov et al.~\cite{bmtw19} and of Huang and Viola~\cite{hv} have established approximately tight results on the maximal statistical distance when $k$ is at most a small constant fraction of $n$ and Naor and Shamir~\cite{NaorS94} gave a tight bound for all $k$ in the case of distinguishing with the $\OR$ function. In this work we provide sharp upper and lower bounds on the maximal statistical distance that holds for all $k$. Upper bounds on the statistical distance have typically been obtained by providing \textit{uniform} low-degree polynomial approximations to certain higher-degree polynomials; the sharpness and wider applicability of our result stems from the construction of suitable non-uniform approximations.

\section{Introduction}
We consider pairs of distributions $\mu$ and $\nu$ over $\{0,1\}^n$. The distributions $\mu$ and $\nu$ are said to be $j$-wise indistinguishable if for any subset $S\subseteq [n]$ of size at most $j$, the the marginal distributions $\mu_S$ and $\nu_S$ over indices in $S$ are identically distributed. The distributions are $k$-wise reconstructible with advantage $\epsilon$ if there exists a set $S\subseteq [n]$ of indices of size $k$ and a statistical test $T:\{0,1\}^{|S|}\rightarrow \{0,1\}$ such that
\[
|\mathbb{E}_{X\sim\mu}[T(X|_S)]-\mathbb{E}_{Y\sim\nu}[T(Y|_S)]|\geq\epsilon,
\]
where $X|_S$ is the restriction of $X$ to the bits located at the indices in $S$. The distributions are symmetric if $\mu$ and $\nu$ are invariant under permutation (see definitions in Section~\ref{s:prelim}); for such distributions the size of $S$ is relevant for distinguishing but not the choice of indices.

\paragraph{Cryptographic motivation.} Work of Bogdanov et al.~\cite{BIVW} considered this notion of indistinguishability as a way to capture cryptographic secret sharing schemes in a minimal setting. Their observation was that a single bit secret can be shared by sampling $n$ bits from $X$ or from $Y$, depending on the secret: the $j$-wise indistinguishability of the distributions provides a security guarantee that any size $\leq j$ coalition of colluding parties learn nothing about the secret from their joint shares. The secret reconstruction function for the scheme is a test $T$ applied over the shares of sufficiently many parties. A key question in their work was how large $j$ could be taken so that there exists a $T$ that is both computable with $AC^0$ circuits and has reconstruction advantage $\epsilon =\Omega (1)$ against some pair of $j$-wise indistinguishable distributions. In this (and other works~\cite{BW17},~\cite{hv}), we consider the statistical distance between the distributions, which includes the study of tests $T$ that are not in $AC^0$ and reconstruction advantage $\epsilon$ that may be vanishing.

\paragraph{Approximate degree motivation.} The work~\cite{BIVW} largely proceeded by a connection to the theory of approximate degree of Boolean functions. The $\epsilon$-approximate degree of a Boolean function $f\colon \{0, 1\}^n \to \R$, denoted $\adeg_{\epsilon}(f)$, is the least degree of a multivariate real-valued polynomial $p$ such that $|p(x)-f(x)| \leq \epsilon$ for all inputs $x \in \zone^n$. This quantity has received significant attention, owing to its polynomial equivalence to many other complexity measures including sensitivity, exact degree, deterministic and randomized query complexity~\cite{NiS94}, and quantum query complexity~\cite{BCdWZ99}. By linear programming duality, $f$ has $\eps$-approximate degree more than $j$ if and only if there exists a pair of probability distributions $\mu$ and $\nu$ over $\zone^n$ such that $\mu$ and $\nu$ are $j$-wise indistinguishable and $2\eps$-distinguishable by $f$. The approximate degree of all symmetric Boolean functions was resolved in the constant-error regime $\epsilon = \Theta (1)$ by Paturi~\cite{paturi1992degree} and in the general error regime by Bun and Thaler~\cite{bun2015dual}. This implies tight upper and lower bounds on the ability of any given symmetric Boolean function to reconstruct from indistinguishable distributions when given access to a full sample of $n$ bits.

\paragraph{Prior work.} Recent works of Bogdanov et al.~\cite{bmtw19} and of Huang and Viola~\cite{hv} extended the study of the indistinguishability of symmetric distributions to the setting of distinguishing with a subset of  indices. They considered the extent to which symmetric $j$-wise indistinguishable distributions must have statistically close $k$-wise marginals for $k>j$. In particular,~\cite{bmtw19} shows that if $\mu$ and $\nu$ are symmetric over $n$-bit strings and perfectly $j$-wise indistinguishable, then the statistical distance between $k$-wise marginals is at most $O(j^{3/2})\cdot e^{-j^2/1156k}$ for all $j < k \leq n/64$. The analogous result in~\cite{hv} gives a similar bound and also applies to $k$ at most some constant fraction of $n$. A matching lower bound given in~\cite{bmtw19} shows that there exists a pair of distributions that are $j$-wise indistinguishable but reconstructable with the $\OR_k$ function with advantage at least $k^{-1/2}\cdot e^{-O(-j^2/k)}$. This lower bound extends to all $j<k\leq n$. Extension of the upper bound to $k$ close to $n$ is of interest due to the fact that the behaviour must change in this parameter range; in particular, the $\PARITY_n$ function can distinguish perfectly between a pair of distributions that are $n-1$ wise indistinguishable. Also, the lower bound for $\OR$ in~\cite{bmtw19} extends for $k$ up to $n$ but is unlikely to come close to matching any upper bound on statistical distance.  

\paragraph{Our contribution.} In the present work, we extend the results of~\cite{bmtw19} and~\cite{hv} to the setting of parameters where $k$ may range freely from 0 to $n-1$ (the case $k=n$ is trivial in light of the $\PARITY$ example above). Our results are ``sharp" in two senses; first, we consider only the sharp threshold reconstruction setting $j=k-1$ and second, our results are tight up to polynomial factors.

\begin{theorem}
\label{thm:mainupper}
There exists an absolute constant $c$ such that for any pair of symmetric $k-1$ wise indistinguishable distributions $\mu,\nu$ over $\{0,1\}^n$, the statistical distance between $\mu|_k$ and $\nu|_k$ is at most:
\[
O(n^c)\cdot \frac{(n-k)^{\frac{n-k}{2}}\cdot (n+k)^{\frac{n+k}{2}}}{2^k\cdot n^n} 
\]
\end{theorem}

\begin{theorem}
\label{thm:mainlower}
There exists a statistical test $T:\{0,1\}^k\rightarrow \{0,1\}$ and a pair of symmetric $k-1$-wise indistinguishable distributions $\mu,\nu$ such that the reconstruction advantage of $T$ is at least
\[
\frac{(n-k)^{\frac{n-k}{2}}\cdot (n+k)^{\frac{n+k}{2}}}{2^k\cdot n^n}.
\] 
\end{theorem}

\paragraph{Techniques and roadmap.}
The established technique to develop indistinguishability upper bounds is to decompose an arbitrary statistical test for symmetric distributions into a small basis using the fact that without loss of generality, the best test is a symmetric function. The basis we work over is $Q_w$ for $w=0,...,k$ where $Q_w$ is a Boolean function that observes $k$ bits and accepts if and only if the observed Hamming weight is exactly $w$. Providing a low-degree polynomial approximation to $Q_w$ rules out the existence of distributions that can be reconstructed with $Q_w$. In practice, a symmetrization is applied to Boolean function $Q_w$ to reduce the problem of its approximation to a problem of approximating a univariate polynomial with a lower degree polynomial. Due to the discrete domain of $Q_w$ (the Boolean cube), the univariate approximations need not be uniform, but are instead over a set of separated points on the real line (for example, $-1,-1+2/n,...,1-2/n,1$).   \\\\
Prior works have not obtained statistical distance upper bounds for $k$ close to $n$ because the approach taken to approximating (the symmetrized version of) $Q_w$ has been to use Chebyshev polynomials to provide uniform approximations to $Q_w$ instead of discrete approximations. This strategy breaks down for large $k$ because the difficulty of uniform approximations diverges from the difficulty of the approximation over the relevant discrete set. This observation motivates the use of discrete Chebyshev polynomials (also known as Gram polynomials) to construct approximations that yield upper bounds on the maximum statistical distance. (We also note that Naor and Shamir~\cite{NaorS94} gave an upper bound for distinguishing $k-1$ wise indistinguishable distributions with the $\OR_k$ function, for all $k$ up to $n$. However, they used an unrelated approximate inclusion-exclusion technique that is tailored to $T=\OR_k$.)\\\\
We provide a lower bound on statistical distance based on hardness of approximation with discrete Chebyshev polynomials. This follows from their orthogonality and from linear programming duality. We believe that prior techniques via orthogonality of (non-discrete) Chebyshev polynomials could be used to show this result (indeed the lower bound result from~\cite{bmtw19} applies to all $k$ up to $n$ and the technique they use should be extendable to distinguishers more powerful than $\OR$).
\\\\After a section on preliminaries, we prove a lemma about the expression of monomials in the discrete Chebyshev basis in Section~\ref{s:monomials}. In Section~\ref{s:moneasy} we construct discrete approximations to the monomial and prove a complementary hardness of approximation result in Section~\ref{s:monhard}. In Section~\ref{s:symmdual} we justify the precise approximation problem to solve using symmetrization and LP duality techniques. Sections~\ref{s:ub} and~\ref{s:lb} justify Theorem~\ref{thm:mainupper} and Theorem~\ref{thm:mainlower}, respectively.

\section{Preliminaries}
\label{s:prelim}
We will be working with polynomial approximations over different discrete sets of points. We define $\dout_n$ as the set of points $\{-1,-1+2/n,...,1\}$ and $\din_n$ as $\{-1+1/n,-1+3/n,...,1-1/n\}$. It is easy to check that $|\dout_n|=n+1$, that $|\din_n|=n$. Further, we have the basic relationships:

\begin{equation}
\label{eq:stretch}
\dout_{n}=\left\{\frac{n+1}{n}\cdot x:x\in\din_{n+1}\right\}
\end{equation}

\begin{equation}
\label{eq:shift}
\din_{n}\subset\left\{x-\frac1n :x\in\dout_n\right\}
\end{equation}

For simplicity we will use $\lesssim$ to hide factors polynomial in $n$.

\paragraph{Symmetric distributions and functions.}
Let $f:\{0,1\}^n\rightarrow\mathbb{R}$ be a function. We say that $f$ is symmetric if the output of $f$ depends only on the Hamming weight of its input. A probability distribution $\mu: \{0,1\}^n\rightarrow [0,1]$ is symmetric if the corresponding probability mass function mapping inputs to probabilities is a symmetric function. 
We also will need two further facts about distinguishing symmetric distributions. Proofs of these appear in~\cite{bmtw19}.
\begin{fact}
\label{fact:symmetricmarginals}
Suppose that $\mu$ is a symmetric distribution over $\{0,1\}^n$. For $S\subseteq \{0,...,n\}$, let $\mu|_{S}$ denote the projection of $\mu$ to the indices in $S$. Then, $\mu|_S$ is also symmetric.
\end{fact}

\begin{fact}
\label{fact:symmetrictest}
Suppose that $\mu$ and $\nu$ are symmetric distributions over $\{0,1\}^n$. Then without loss of generality, the best statistical test $Q:\{0,1\}^n\rightarrow [0,1]$ for distinguishing between $\mu$ and $\nu$ is a symmetric function. In particular, we have:
\[\max_{\text{symmetric }Q}\{\E_\mu[Q(\mu)]-\E_\nu[Q(\nu)]\}=\max_{Q}\{\E_\mu[Q(\mu)]-\E_\nu[Q(\nu)]\}.\]
\end{fact}

\subsection{Discrete Chebyshev polynomials}
The discrete Chebyshev polynomials (sometimes called the Gram polynomials), for parameter $n$ are a family of real polynomials $\{\phi_d\}_{d=0,...,n-1}$. Borrowing notation from~\cite{barnard1998gram}, we have that the polynomials have the following properties:
\begin{itemize}
\item The family of polynomials $\{\phi_d\}_{d=0,...,n-1}$ are orthogonal with respect to the bilinear form given by

\begin{equation}
(\phi_i,\phi_j):=\frac1n\cdot\sum_{x\in\din_n}\phi_i(x)\cdot \phi_j(x)
\end{equation} 
\item For each $d$:

\begin{equation}
\label{eq:orthfact}
||\phi_d||:=(\phi_d,\phi_d)^{1/2}=1
\end{equation}
\item For each $d$: \begin{equation}\text{deg}(\phi_d)=d\end{equation}
\item The polynomials satisfy the recurrence:

\begin{equation}
\label{eq:6}
\phi_d(x)=2\alpha_{d-1}\cdot x\cdot \phi_{d-2}(x)-\frac{\alpha_{d-1}}{\alpha_{d-2}}\cdot \phi_{d-2}(x)
\end{equation}

\begin{equation}
\label{eq:7}
\alpha_{d-1}=\frac{n}{d}\cdot\left(\frac{d^2-1/4}{n^2-d^2}\right)^{1/2}
\end{equation}
where we have $\phi_0=1$, $\phi_{-1}=0$, and $\alpha_{-1}=1$.
\end{itemize}

Every degree-$k$ polynomial $p\colon \R \to \R$ has a unique expansion in the Gram basis:
\[ p(t) = \sum_{d = 0}^k c_d \phi_d(t),  \] 
where $c_{0}, \dots, c_K$ are the {\em Gram coefficients} of $p$.  

\subsection{Bounds on factorials and binomial coefficients}
We will make use of double factorials, which are given by:

\begin{equation}
\label{eq:doublefactorial}
n!!:=\prod_{i=0}^{\lfloor n/2\rfloor}n-2i
\end{equation}
and satisfy, when $n$ is even:

\begin{equation}
\label{eq:df}
n!!=2^{n/2}\cdot (n/2)!
\end{equation}
For $n$ odd, we simply observe that $(n-1)!! \lesssim n!!\lesssim (n+1)!!$ and apply the bound in (\ref{eq:df}).

We will bound factorials using

\begin{equation}
\label{eq:10}
n!=\Theta (\sqrt{n})\cdot \left(\frac{n}{e}\right)^n
\end{equation}
and the central binomial coefficient using

\begin{equation}
\label{eq:11}
\binom{2n}{n}=\Theta (1/\sqrt{n})\cdot 2^{2n}
\end{equation}

\subsection{Approximate degree of Boolean functions}
Let $f:\{0,1\}^n\rightarrow \{0,1\}$ be a Boolean function. We will use $\adeg_{\eps}(f)$ to denote the minimum degree of any real polynomial $p:\{0,1\}^n\rightarrow \mathbb{R}$ that approximates $f$ to within $\epsilon$ at every point in $\{0,1\}^n$.

\section{Monomials in the Gram basis}
\label{s:monomials}
\begin{lemma}
\label{lem:monomialgram}
Fix integer $k<n$. Let $C$ be the leading coefficient in the expansion of $x^k$ in the Gram polynomial basis with parameter $n$. Then, $C$ satisfies:
\[
\frac{2}{(2n)^k}\cdot\sqrt{\frac{(n+k)!}{(n-k)!}} \leq C \lesssim \frac{1}{(2n)^k}\cdot\sqrt{\frac{(n+k)!}{(n-k)!}} 
\]
\end{lemma}
\begin{proof}
From the recurrence definition of the Gram polynomials in Equation~\ref{eq:6}, we have that
\[
x\cdot\phi_i=\frac{1}{2\alpha_{i-1}}\phi_{i+1}+\frac{1}{2\alpha_{i-2}}\phi_{i-1}
\]
Application of this, along with the base case $\phi_0=1$, yields that the highest degree term of $x^k$ when expressed in the Gram basis is given by 
\begin{equation}
\label{e:eq}
C=\prod_{i=0}^k\frac{1}{2\alpha_i}=2^{-(k+1)}\cdot\prod_{i=0}^k\frac{1}{\alpha_i}
\end{equation}

From the definition of the $\alpha_i$ in Equation~\ref{eq:7}, we have that

\begin{align*} 
\prod_{i=0}^k\alpha_i &= \frac{n^k}{k!}\sqrt{\frac{(1^2-1/4)(2^2-1/4)\cdot ...\cdot (k^2-1/4)}{(n^2-1^2)(n^2-2^2)\cdot ...\cdot(n^2-k^2)}}\\
  &= \frac{n^k}{k!}\sqrt{\frac{(1^2-1/4)(2^2-1/4)\cdot ...\cdot (k^2-1/4)}{(n+1)(n-1)(n+2)(n-2)\cdot ...\cdot(n+k)(n-k)}}\\
    &= n^k\cdot\sqrt{\frac{(n-k)!}{(n+k)!}}\cdot\sqrt{\frac{(1^2-1/4)(2^2-1/4)\cdot ...\cdot (k^2-1/4)}{k^2\cdot (k-1)^2\cdot ...\cdot 1^2}}\\
    &= n^k\cdot\sqrt{\frac{(n-k)!}{(n+k)!}}\cdot \prod_{i=1}^k\frac{(i-1/2)(i+1/2)}{i^2}.
            \end{align*}
Application of the upper bound in Claim~\ref{c:technical} yields that $\prod_{i=0}^k\alpha_i\leq n^k\cdot\sqrt{\frac{(n-k)!}{(n+k)!}}$. This, in conjunction with Equation~\ref{e:eq}, justifies the lower bound in the statement of this lemma. For the upper bound, application of the lower bound in Claim~\ref{c:technical} yields that $\prod_{i=0}^k\alpha_i\geq \frac{1}{k^2}\cdot n^k\cdot\sqrt{\frac{(n-k)!}{(n+k)!}}$, which completes the proof in light of Equation~\ref{e:eq}.
\end{proof}

\section{Discrete approximations for the monomial}
\label{s:moneasy}

\begin{cor}
\label{cor:maineasy}
Fix integers $k<n$. There exists a degree at most $k-1$ polynomial approximation for the monomial $x^k$ over $\dout_n$ with error $\epsilon$ satisfying:
\[
\epsilon\lesssim  (2(n+1))^{-k}\cdot\sqrt{\frac{(n+k+1)!}{(n-k+1)!}}
\]
\end{cor}

\begin{proof}
It suffices to provide a degree at most $k-1$ approximation $p$ to the monomial $(\frac{n+1}{n}\cdot x)^k$ over $\din_{n+1}$ because then the degree at most $k-1$ approximation $p':=p(\frac{n}{n+1}\cdot x)$ will be an approximation to $x^k$ over $\dout_n$, by Equation~\ref{eq:stretch}. By Lemma~\ref{lem:monomialgram}, the expansion of $(\frac{n+1}{n}\cdot x)^k$ in the Gram basis with parameter $n+1$ is given by
\[
C_k\cdot \phi_k+C_{k-1}\cdot\phi_{k-1}+...+C_0\cdot\phi_0, 
\]
where $C_k\lesssim \left(\frac{n+1}{n}\right)^k\cdot \frac{1}{(2(n+1))^k}\cdot\sqrt{\frac{(n+k+1)!}{(n-k+1)!}}\lesssim \frac{1}{(2(n+1))^k}\cdot\sqrt{\frac{(n+k+1)!}{(n-k+1)!}}$. By Equation~\ref{eq:orthfact} and Cauchy-Schwarz, $\max_{\din_{n+1}}\phi_k\lesssim 1$ and the corollary follows by taking the approximation $\sum_{i=0}^{k-1}C_i\cdot\phi_i$. 
\end{proof}

\paragraph{Comparison to the uniform approach.} Newman and Rivlin~\cite{newman1976approximation}, and Sachdeva and Vishnoi~\cite{sachdeva2013approximation} showed that any uniform approximation to the monomial $x^k$ over $[-1,1]$ will have error $2^{-k+1}$ (and that this is tight). For large enough values of $n$ and $k$, the bound in the statement of Corollary~\ref{cor:maineasy} is smaller. In Section~\ref{s:ub} we see that the bound $2^{-k+1}$ is insufficient to get a non-trivial indistinguishability upper bound for $k$ close to $n$.

\section{Hardness of discrete monomial approximations}
\label{s:monhard}
\begin{cor}
\label{cor:mainhard}
Fix integers $k<n$. Any degree at most $k-1$ polynomial approximation to the monomial $x^k$ must have error over $\dout_n$ at least:
\[
(2n)^{-k}\cdot\sqrt{\frac{(n+k)!}{(n-k)!}}
\]
\end{cor}

The proof of the main corollary of this section appears at the end of this section after two lemmas have been established.

\begin{lemma}
\label{lem:gramorth}
Let $p$ be a degree $k$ polynomial with degree $k$ coefficient $C$ in the Gram basis with parameter $n$. Then, any degree $k-1$ approximating polynomial will have error at least $|C|$ at some point over $\din_n$.
\end{lemma}
\begin{proof}
Let $q$ be any degree at most $k-1$ polynomial. Let $c_i$ be the degree $i$ coefficient in the Gram representation of $p-q$, and note that because the degree of $q$ is at most $k-1$, we have that $c_k=C$. By orthogonality of the Gram polynomials, and Equation~\ref{eq:orthfact},
 \[
\E_{t \sim \din_n }[(p(t) - q(t))^2] 
  = c_0^2 + \sum_{d =1}^k (c_d)^2 \E_{t \sim \din_n}[\phi_d(t)^2]
  \geq c_k^2=C^2.
\]
It follows that the approximation error $\abs{p(t) - q(t)}$ must exceed $|C|$ for some $t \in \din_n$.
\end{proof}

\begin{lemma}
\label{lem:doutdin}
Let $p$ be a degree $k$ polynomial such that for any degree at most $k-1$ polynomial $q$, $\max_{t\in\din_n}\{|p(t)-q(t)|\}\geq\epsilon$. Then, for any degree $k-1$ polynomial $q'$, $\max_{t\in\dout_n}\{|p(t)-q'(t)|\}\geq\epsilon$. 
\end{lemma}

\begin{proof}
We consider the contrapositive and show that existence of a degree at most $k-1$ polynomial $\tilde{p}$ for $p$ over $\dout_n$ with error at most $\epsilon$ implies an approximation for $p$ over $\din_n$ with the same degree and error parameters. We have that $\tilde{p}(x+1/n)$ is a degree $k-1$ approximation of $p(1+1/n)$ over $\{x-\frac1n : x\in\dout_n\}\supset \din_n$, where the set relation follows from Equation~\ref{eq:shift}. Our approximation is then $\tilde{p}(t+1/n)+p(t)-p(t+1/n)$, which is degree $k-1$ because $p(t)-p(t+1/n)$ is degree $k-1$. We have that $\max_{t\in \din_n}|\left(\tilde{p}(t+1/n)+p(t)-p(t+1/n)\right)-p(t)|=\max_{t\in \din_n}|\tilde{p}(t+1/n)-p(t+1/n)|\leq\epsilon$.
\end{proof}

\begin{proof}[Proof of Corollary~\ref{cor:maineasy}]
Lemma~\ref{lem:monomialgram} and Lemma~\ref{lem:gramorth} imply that any degree at most $k-1$ polynomial approximation to $x^k$ must have error at least $(2n)^{-k}\cdot\sqrt{\frac{(n+k)!}{(n-k)!}}$ over $\din_n$. The corollary then follows from Lemma~\ref{lem:doutdin}.
\end{proof}

\section{Symmetrization and duality}
\label{s:symmdual}

Our main upper and lower bounds will be justified by reducing to an approximation theoretic question using a linear programming duality relation.

\begin{claim}[see, for example,~\cite{BIVW}]
\label{c:dualityclaim}
$\adeg_{\eps/2}(F) \geq k$ if and only if there exists a pair of perfectly $k$-wise indistinguishable distributions $\mu$, $\nu$ over $\{0,1\}^n$ such that $\E_{x \sim \mu}[F(x)] - \E_{y \sim \nu}[F(y)] \geq \eps$.
\end{claim}

We are interested in Boolean functions as statistical tests that witness $k$ bits of a sample from a distribution. To this end, let $Q_w$ denote the function on $\{0, 1\}^n$ that outputs 1 if and only if the Hamming weight of $x|_{\{1, \dots, k\}}$, the first $k$ bits of the input, is exactly $w$. It will be easier for us to work over functions that are symmetrised, i.e. not a $k$-junta.

\begin{fact}
\label{fact:pw}
Let $S\subseteq [n]$ be any set of size $k$. There exists a univariate polynomial $p_w$ of degree at most $k$ such that
the following holds.
For all $t \in \dout_n$, $p_w(t) = \E_{Z}[Q_w(Z|_S)]$ where $Z$ is a random string of Hamming weight $\phi^{-1}(t) = (1 - t)n/2 \in \{0, 1, \dots, n\}$.
\end{fact}
\begin{proof}
This statement is a simple extension of Minsky and Papert's classic symmetrization technique \cite{MiP69} and also appears in~\cite{bmtw19}; we reproduce the proof here for convenience. Minsky and Papert showed that for any polynomial $P \colon \{0, 1\}^n \to \R$,
there exists a univariate polynomial $p$ of degree at most the total degree of $P$, such
that for all $i \in \{0, \dots, n\}$, $p(i) = \mathbb{E}_{|x|=i}[P(x)]$. 
Apply this result to $P(x) = Q_w(x|_S)$ and let $p_w(t) = p(\phi^{-1}(t)) = p\left((1 - t)n/2\right)$.
The fact then follows from the observation that the total degree of $Q_w(x|_S)$ is at most $k$,
since this function is a $k$-junta. 
\end{proof}

\begin{cor}
\label{cor:multitodout}
Suppose that for any degree $\leq k-1$ polynomial $q$ we have that $\max_{t\in\dout_n}\{|p_w(t)-q|\}\geq\epsilon/2$. Then, $\adeg_{\eps/2}(Q_w(x|_S)) \geq k-1$.
\end{cor}

\begin{proof}
This is immediate from the contrapositive of Fact~\ref{fact:pw}.
\end{proof}

\subsection{Properties of $p_w$}
The value $p_w(t)$ is a probability for every $t \in \dout_n$. 
Moreover, this probability must equal zero when the Hamming weight of $Z$ is less than $w$ or greater than $n - k + w$.  Therefore $p_w$ has $k$ distinct zeros at the points $Z_w = Z_{-} \cup Z_{+}$, where
\begin{equation}
\label{eq:zeros}
Z_{-}=\left\{-1+ 2h/n: h=0,...,k-w-1\right\}, \\ Z_{+} = \{1 - 2h/n: h=0,...,w-1\}.
\end{equation}
and so $p_w$ must have the form
\begin{equation}
\label{eq:pw}
p_w(t) = C_w \cdot \prod_{z \in Z_w} (t - z)
\end{equation}
for some $C_w$ that does not depend on $t$. 

\begin{claim}
\label{claim:leadpwterm}
The coefficient on the highest degree term of $p_w$ in the monomial basis is $C_w$, which equals:
\[
\frac{\binom{k}{w}\binom{n-k}{\frac12 (n-k)}}{\binom{n}{\frac12 (n-k+2w)}}\cdot \frac{n^{k}\cdot (n-k)!!^2}{(n-k+2w)!!\cdot (n-2w+k)!!}
\]
\end{claim}

\begin{proof}
The polynomial $p_w$ is of degree $k$ with all of its zeroes lying in $Z_w$. We evaluate $p_w$ at a point $t'=\frac{k-2w}{n}$ which is necessarily outside of $Z_w$ and thus not a zero of $p_w$. To do this, we use that the value $p_w(t')$ is the probability that $Q_w(x|_S)$ accepts given that $x$ is chosen uniformly at random, conditioned on the event that the Hamming weight of $x$ is exactly $\phi^{-1}(t')=\frac12 (n-k+2w)$. 
\[
\Pr\left[Q_w(x|_S)=1:|x|=\frac12 (n-k+2w)\right]=p_w(t')=C_w\cdot \prod_{z\in Z_w}(t'-z),
\]
from which it follows that 
\[
C_w=\frac{\binom{k}{w}\binom{n-k}{\frac12 (n-k)}}{\binom{n}{\frac12 (n-k+2w)}\cdot\prod_{z\in Z_w}(t'-z)}
\]
We have that:
\begin{align*} 
\prod_{z\in Z_w}(t'-z) &= \prod_{z\in Z_-}(t'-z)\prod_{z\in Z_+}(t'-z)\\
  &= \prod_{z\in Z_+}(z-t')^2\prod_{z\in Z_-:-z\not\in Z_+}(t'-z),
            \end{align*}
where the final equality assumes that $w\leq k/2$. This is without loss of generality; when $w>k/2$, the same calculation holds with the roles of $Z_+$ and $Z_-$ reversed. From this we compute that:

\begin{align*} 
\frac{1}{\prod_{z\in Z_w}(t'-z)} &= \frac{1}{(1-\frac{k-2}{n})^2\cdot (1-\frac{k-2}{n}+\frac2n)^2\cdot ...\cdot (1-\frac{k-2}{n}+\frac{2(w-1)}{n})^2\cdot \prod_{z\in Z_-:-z\not\in Z_+}(t'-z)}\\
  &= \frac{n^{2w}\cdot (n-k)!!^2}{(n-k+2w)!!^2\cdot \prod_{z\in Z_-:-z\not\in Z_+}(t'-z)}\\
  &= \frac{n^{2w}\cdot (n-k)!!^2}{(n-k+2w)!!^2\cdot \prod_{i=0}^{K-2w-1}(\frac{k-2w}{n}+1-\frac{2i}{n})}\\
    &= \frac{n^{k}\cdot (n-k)!!^2}{(n-k+2w)!!\cdot (n+k-2w)!!},
            \end{align*}
from which the claim follows.
\end{proof}

We also will need the following fact, which is justified in the Appendix. 
\begin{lemma}
\label{lem:argmax}
The value $C_w$ is maximized when $w=k/2$; in particular with 
\[
C_{k/2}=\frac{\binom{k}{k/2}\binom{n-k}{(n-k)/2}}{\binom{n}{n/2}}\cdot \frac{n^{k}\cdot (n-k)!!^2}{n!!^2}
\]
\end{lemma}

\section{Upper bound}
\label{s:ub}
\begin{lemma}
\label{lem:sneaky}
For any $w=0,...,k$ and pair of $k-1$ wise indistinguishable distributions, the function $Q_{w}$ reconstructs with advantage $\epsilon$, satisfying:
\[
\epsilon\lesssim \frac{(n-k)^{\frac{n-k}{2}}\cdot (n+k)^{\frac{n+k}{2}}}{2^k\cdot n^n} 
\]
\end{lemma}
\begin{proof}
By Corollary~\ref{cor:maineasy}, there exists a degree $k-1$ polynomial approximation to $p_w$ over $\dout_n$ with error
\begin{equation}
\label{e:compare2}
\lesssim C_w\cdot\frac{1}{(2(n+1))^k}\cdot\sqrt{\frac{(n+k+1)!}{(n-k+1)!}}
\end{equation}
By Claim~\ref{claim:leadpwterm} and Lemma~\ref{lem:argmax} this is upper bounded by (up to $\text{poly}(n)$ factors):
\[
\frac{\binom{k}{k/2}\binom{n-k}{(n-k)/2}}{\binom{n}{n/2}}\cdot \frac{n^{k}\cdot (n-k)!!^2}{n!!^2}\cdot\frac{1}{(2(n+1))^k}\cdot\sqrt{\frac{(n+k+1)!}{(n-k+1)!}}
\]

\begin{align*}
\lesssim &  \frac{k!\cdot (n-k)!\cdot n^k}{(k/2)!^2\cdot n!\cdot 2^k}\cdot \frac{1}{(2n+2)^k}\cdot  \sqrt{\frac{(n+k+1)!}{(n-k+1)!}}\\
\lesssim & \frac{(n-k)^{n-k}\cdot (2e)^k}{n^n}\cdot\frac{n^k}{2^k} \cdot \frac{1}{(2n+2)^k}\cdot  \frac{(n+k+1)^{(n+k+1)/2)}}{e^k\cdot (n-k+1)^{(n-k+1)/2}} \\
\lesssim & \frac{(n-k)^{(n-k-1)/2}\cdot (n+k+1)^{(n+k+1)/2}}{2^k\cdot n^n} \\
\lesssim & \frac{(n-k)^{(n-k)/2}\cdot (n+k)^{(n+k)/2}}{2^k\cdot n^n},
\end{align*}
where we have used Equations~\ref{eq:11},~\ref{eq:10}, and~\ref{eq:df} to bound the central binomial coefficient, the factorial, and the double factorial, respectively.
\paragraph{Comparison to uniform approach.} We saw in Section~\ref{s:moneasy} that any \textit{uniform} approximation to the monomial would have error $2^{-k+1}$. Substituting that bound into Equation~\ref{e:compare2} and carrying out the same calculation would yield an upper bound of $\frac{e^k\cdot (n-k)^{n-k}}{2^k\cdot n^{n-k}}$, which for $k\approx n$ is $\gtrsim (e/2)^n$. Because any distinguishing advantage must be at most 1, this is a vacuous bound.

\end{proof}

\subsection{Proof of upper bound: Theorem~\ref{thm:mainupper}}
\begin{theorem}
For any pair of $k-1$ wise indistinguishable distributions $\mu,\nu$ over $\{0,1\}^n$, the statistical distance $\epsilon$ between $\mu|_k$ and $\nu|_k$ satisfies:
\[
\epsilon \lesssim \frac{ (n-k)^{\frac{n-k}{2}}\cdot (n+k)^{\frac{n+k}{2}}}{2^k\cdot n^n} 
\]
\end{theorem}
\begin{proof}
Let $T$ be a general distinguisher on $k$ inputs.  By Facts~\ref{fact:symmetricmarginals} and~\ref{fact:symmetrictest}, $T$ can be assumed to be a symmetric Boolean-valued function and has the representation $T=\sum_{w=0}^k b_w\cdot Q_w$ where each of the $b_w$ is either 0 or 1. We bound the distinguishing advantage as follows. Recalling that $\mu$ and $\nu$ are $k-1$-indistinguishable symmetric distributions over $\{0, 1\}^n$, for any set $S \subseteq [n]$ of size $k$ we have:

\begin{align*}
\E[T(\mu|_S)] - \E[T(\nu|_S)]  
 &= \sum_{w = 0}^k b_w \bigl(\E[Q_w(\mu|_S)] - \E[Q_w(\nu|_S)]\bigr)  \\
 &\leq \sum_{w=0}^k \bigabs{\E[Q_w(\mu|_S)] - \E[Q_w(\nu|_S)]}   \\
   &\leq (k+1) \cdot\max_{w=0,...,k} \bigabs{\E[p_{w}(\phi(\abs{\mu})] - \E[p_{w}(\phi(\abs{\nu}))]}  \\
   &\lesssim \frac{(n-k)^{\frac{n-k}{2}}\cdot (n+k)^{\frac{n+k}{2}}}{2^k\cdot n^n},
\end{align*}
where the final upper bound is from Lemma~\ref{lem:sneaky}.
\end{proof}

\section{Proof of lower bound: Theorem~\ref{thm:mainlower}}
\label{s:lb}
\begin{theorem}
Let $Q_{k/2}$ be the statistical test over $k$ bits that accepts if and only if the observed Hamming weight is $k/2$. There exists a pair of $k-1$-wise indistinguishable distributions $X,Y$ such that the reconstruction advantage of $Q_{k/2}$ is at least $\epsilon$, satisfying:
\[
\epsilon \gtrsim \frac{(n-k)^{\frac{n-k}{2}}\cdot (n+k)^{\frac{n+k}{2}}}{2^k\cdot n^n}.
\] 
\end{theorem}
\begin{proof}
By Claim~\ref{c:dualityclaim} and Corollary~\ref{cor:multitodout}, it suffices to show that any degree $k-1$ polynomial approximation to $p_{k/2}$ over $\dout_n$ must have error at least $\epsilon$. Lemma~\ref{lem:doutdin} reduces the problem further to proving hardness of approximation of $p_{k/2}$ over $\din_n$. From Claim~\ref{claim:leadpwterm} and Lemma~\ref{lem:monomialgram}, the coefficient on the degree $k$ term of the Gram representation of $p_{k/2}$ is
\[
\gtrsim \frac{\binom{k}{k/2}\binom{n-k}{(n-k)/2}}{\binom{n}{n/2}}\cdot \frac{n^{k}\cdot (n-k)!!^2}{n!!^2}\cdot \frac{1}{(2n)^k}\cdot\sqrt{\frac{(n+k)!}{(n-k)!}},
\]
which is $\gtrsim \epsilon$, by trivially applying the bounds for the central binomial coefficient, factorial, and double factorial in Section~\ref{s:prelim}. The theorem then follows by applying Lemma~\ref{lem:gramorth}.
\end{proof}

\newpage
\bibliographystyle{abbrv}
\bibliography{OmniBib}

\appendix
\section{A technical claim}
\begin{claim}
\label{c:technical}
Let $v=\prod_{i=1}^k\frac{(i-1/2)(i+1/2)}{i^2}$. Then, 
\[
\frac{1}{k^2} \leq v \leq 1
\]
\end{claim}
\begin{proof}
We have that $v=\prod_{i=1}^k\left(1-\frac{1}{4i^2}\right)$, which is a product of numbers less than 1 and justifies the upper bound. For the lower bound, we have:
\begin{align*} 
\frac1v &=\frac{1}{(1-1/2^2)(1-1/3^2)\cdot ...\cdot (1-1/k^2)}\\
    &= \frac{2^2\cdot 3^2\cdot ... \cdot k^2}{(2^2-1)\cdot (3^2-1)\cdot ... \cdot (k^2-1)}\\
        &\leq \frac{2^2\cdot 3^2\cdot ... \cdot k^2}{1^2\cdot 2^2\cdot ... \cdot (k-1)^2}= k^2.
                \end{align*}
\end{proof}

\section{Proof of Lemma~\ref{lem:argmax}}
\begin{proof} We find the maximising value of $C_w$ by expanding the expression for $C_w$ and removing terms that do not depend on $w$:
\begin{align*} 
\arg\max_{w}C_w&= \arg\max_w \frac{\binom{k}{w}\binom{n-k}{\frac12 (n-k)}}{\binom{n}{\frac12 (n-k+2w)}}\cdot \frac{n^{k}\cdot (n-k)!!^2}{(n-k+2w)!!\cdot (n-2w+k)!!}\\
  &= \arg\max_w  \frac{\binom{k}{w}}{\binom{n}{\frac12 (n-k+2w)}}\cdot \frac{1}{(n-k+2w)!!\cdot (n-2w+k)!!}\\
    &= \arg\max_w \frac{k!\cdot (\frac{n-k+2w}{2})!\cdot (n-\frac{n-k+2w}{2})!}{w!\cdot (k-w)!\cdot n! \cdot (n-k+2w)!!\cdot (n-2w+k)!!}\\
        &= \arg\max_w  \frac{(\frac{n-k+2w}{2})!\cdot (\frac{n+k-2w}{2})!}{w!\cdot (k-w)!\cdot (n-k+2w)!!\cdot (n-2w+k)!!}\\
                &= \arg\max_w  \frac{(\frac{n-k+2w}{2})!\cdot (\frac{n+k-2w}{2})!}{w!\cdot (k-w)!\cdot 2^n\cdot (\frac{n-k+2w}{2})!\cdot (\frac{n-2w+k}{2})!}\\
                                &= \arg\max_w  \frac{1}{w!\cdot (k-w)!}=k/2
                \end{align*}
\end{proof}

\end{document}